\documentclass[a4paper,UKenglish]{article}

\usepackage{fullpage}
\usepackage{amsmath}
\usepackage{amsthm}
\usepackage{amssymb}
\usepackage{authblk}
\usepackage{breqn}
\usepackage{array}
\usepackage{enumerate}
\usepackage{algorithm}
\usepackage{algpseudocode}
\usepackage{bookmark}
\usepackage{hyperref}
\hypersetup{
  pdftitle = {Exact quantum query complexity of EXACT(n,k,l)},
  pdfauthor = {},
  pdfkeywords = {},
  pdfstartview=FitH,
  unicode=true
}

\newtheorem{theorem}{Theorem}
\newtheorem{definition}{Definition}
\newtheorem{conjecture}{Conjecture}
\newtheorem{lemma}{Lemma}

\newcommand{\ket}[1]{\left| #1 \right\rangle}

\newcommand{\f}[1]{\textsc{#1}}
\newcommand{\sr}[1]{\texttt{#1}}
\newcommand{\comment}[1]{}
\newcommand{\dnote}[1]{}
\newcommand{\jnote}[1]{}
\newcommand{\email}[1]{\href{mailto:#1}{\texttt{#1}}}

\def\H{{\cal H}}
\def\L{{\cal L}}
\def\R{{\cal R}}
\def\S{{\cal S}}

\title{Exact quantum query complexity of $\f{EXACT}_{k,l}^n$}
\author[1]{Andris Ambainis}
\author[1]{J\=anis Iraids}
\author[2]{Daniel Nagaj}
\affil[1]{Faculty of Computing, University of Latvia, Rai\c{n}a bulv\=aris 19, Riga, LV-1586, Latvia, \email{ambainis@lu.lv}, \email{janis.iraids@gmail.com}}
\affil[2]{Institute of Physics, Slovak Academy of Sciences, D\'{u}bravsk\'{a} cesta 9, 845 11 Bratislava, Slovakia,\email{dnagaj@gmail.com}}

\begin{document}
\maketitle

\begin{abstract}
\dnote{A variant of an abstract that is just a short text: 

In the exact quantum query model a successful algorithm must always output the correct function value. We investigate the function that is true if exactly $k$ or $l$ of the $n$ input bits given by an oracle are 1. We find an optimal algorithm (for some cases), and a nontrivial general lower and upper bound on the minimum number of queries to the black box. 
}

In the decision tree model, one's task is to compute a boolean function $f:\{0,1\}^n\to\{0,1\}$ on an input $x\in\{0,1\}^n$ that is accessible via queries to a black box (the black box hides the bits $x_i$). \dnote{I added the parentheses} 
In the quantum case, classical queries and computation are replaced by unitary transformations. 
A quantum algorithm is exact if it always outputs the correct value of $f$ (in contrast to the standard model of 
quantum algorithms where the algorithm is allowed to be incorrect with a small probability). The minimum number of queries for an exact quantum algorithm computing the function $f$ is denoted by $Q_E(f)$.

We consider the following $n$ bit function with $0\leq k \leq l \leq n$:
\[\f{EXACT}_{k,l}^n(x)=\begin{cases}1\text{, if }x_1+\ldots+x_n\in\{k,l\},\\
    0\text{, otherwise},\end{cases}\]
i.e. we want to give the answer 1 only when exactly $k$ or $l$ of the bits $x_i$ are 1.
We construct a quantum query algorithm for this function and give lower bounds for it, with lower bounds matching the complexity of the algorithm in some cases (and almost matching it in other cases):
\dnote{Please, reread, I rewrote this a little.}
\begin{itemize}
   \item If $l-k=1$ and $k=n-l$, then $Q_E(\f{EXACT}_{k,k+1}^{2k+1})=k+1$.
   \item If $l-k\in\{2,3\}$, then $Q_E(\f{EXACT}_{k,l}^n)=\max\{n-k,l\}-1$.
   \item For all $k,l$: $\max\{n-k, l\}-1 \leq Q_E(\f{EXACT}_{k,l}^n) \leq \max\{n-k, l\}+1$.
\end{itemize}
\end{abstract}


\section{Introduction}
In this paper we study the computational complexity of boolean functions in the quantum black box model. It is a generalization of the decision tree model, where we are computing an $n$-bit function $f:\{0,1\}^n\to \{0,1\}$ on an input $x\in\{0,1\}^n$ that can only be accessed through a black box by querying some bit $x_i$ of the input. 
In the quantum black box model the state of the computation is described by a quantum state from the Hilbert space $\H_Q\otimes \H_W \otimes \H_O$ where $\H_Q=\{\ket{0},\ket{1}, \ldots, \ket{n}\}$ is the query subspace, $\H_W$ is the working memory and $\H_O=\{\ket{0},\ket{1}\}$ is the output subspace. A computation using $t$ queries consists of a sequence of unitary transformations $U_t\cdot O_x \cdot U_{t-1} \cdot O_x \cdot \ldots \cdot O_x \cdot U_0$ followed by a measurement, where the $U_i$'s are independent of the input and $O_x=O_{Q,x}\otimes I\otimes I$ with
\[O_{Q,x}\ket{i}=\begin{cases}(-1)^{x_i}\ket{i}=\hat{x}_i \ket{i}\text{, if }i\in[n],\\
\ket{0}\text{, if }i=0,
\end{cases}
\]
is the query transformation, where $x_i \in \{0,1\}$ or equivalently, $\hat{x}_i \in \{-1,1\}$.
The final measurement is a complete projective measurement in the computational basis and the output of the algorithm is the result of the last register, $\H_O$. We say that a quantum algorithm computes $f$ exactly if for all inputs $x$ the output of the algorithm always equals $f(x)$. Let us denote by $Q_E(f)$ the minimum number of queries over all quantum algorithms that compute $f$ exactly.

For quite a long time the largest known separation between the classical decision tree complexity $D(f)$ and $Q_E(f)$ was only by a factor of two --- the $\f{XOR}$ of two bits can be computed exactly using only 1 quantum query\cite{CEMM98,DJ92,FGGS98}. However, in 2012 Ambainis gave the first asymptotic separation that achieved $Q_E(f)=O(D(f)^{0.8675})$ for a class of functions $f$ \cite{Amb13}. Next, in 2015 Ambainis et al. used pointer functions
to show a near quadratic separation between these two measures: $Q_E(f)=\tilde{O}(\sqrt{D(f)})$ \cite{ABB+15}. On the other hand Midrij\={a}nis has proved that the maximum possible separation between $Q_E(f)$ and $D(f)$ is at most cubic \cite{Mid04}.

However, 
the techniques for designing exact quantum algorithms are rudimentary compared to the bounded error setting. Other than the well known {\em $\f{XOR}$ trick} --- constructing a quantum algorithm from a classical decision tree that is allowed to ``query'' the $\f{XOR}$ of any two bits --- there are few alternate approaches. In addition to the asymptotic separations of \cite{ABB+15,Amb13}, Montanaro et al. \cite{MJM15} gave a 2 query quantum algorithm for the symmetric 4 bit function
\[
    \f{EXACT}_2^4(x)
    =\begin{cases}
        1\text{, if }x_1+x_2+x_3+x_4=2,\\
        0\text{, otherwise,}\end{cases}
\] 
and showed that it could not be computed optimally using the XOR trick. Afterwards Ambainis et al. gave an algorithm \cite{AIS13} for two classes of symmetric functions:
\[
    \f{EXACT}_k^n(x)
    =\begin{cases}
        1\text{, if }x_1+x_2+\ldots+x_n=k,\\
        0\text{, otherwise}
    \end{cases};
    \quad Q_E(\f{EXACT}_k^n)\leq \max\{k,n-k\},
\]
and the threshold function
\[
    \f{TH}_k^n(x)
    =\begin{cases}
        1\text{, if }x_1+x_2+\ldots+x_n\geq k,\\
        0\text{, otherwise}
    \end{cases};
    \quad Q_E(\f{TH}_k^n)\leq \max\{k,n-k+1\}.
\]
For partial functions quantum algorithms with superpolynomial speedup are known \cite{DJ92,BH97}. It seems that our work relates well to the results of Qiu and Zheng on partial functions based on the Deutsch-Jozsa problem \cite{QZ16}.


\subsection{Our results}
\label{sec:ourresults}

We consider exact quantum algorithms for symmetric total boolean functions, i.e., functions for which permuting the input bits does not change its value. For symmetric functions, the largest known separation remains a factor of 2. We know from von zur Gathen's and Roche's work on polynomials \cite{vZGR97} and quantum lower bounds using polynomials \cite{BBC+98} that for symmetric $f: Q_E(f)\geq \frac{n}{2}-O(n^{0.548})$, thus the largest known separation is either optimal or close to being optimal.

However, many of the known exact algorithms are for symmetric functions (for example, XOR, EXACT and TH functions mentioned in the previous section). Because of that, we think that symmetric functions may be an interesting ground to explore new methods for developing more exact quantum algorithms.

In Section~\ref{sec:unb3} we present an algorithm achieving up to $D(f)=2Q_E(f)$ for a certain class of symmetric functions 
\begin{definition}
  Let $\f{EXACT}_{k,l}^n$ for $0 \leq k \leq l \leq n$, be an $n$-argument symmetric boolean function that returns 1 if and only if the input contains exactly $k$ ones or exactly $l$ ones.
  \[\f{EXACT}_{k,l}^n(x)=\begin{cases}1,&\text{ if }|x|\in\{k,l\};\\
  0,&\text{ otherwise.}
  \end{cases}\]
\end{definition}
Let us denote by $d$ the separation between $l$ and $k$: $d=l-k$. In general a symmetric boolean function $\f{SYM}_a$ on $n$ input bits can be defined by a list $a=(a_0,\ldots, a_n)\in\{0,1\}^{n+1}$ such that $\f{SYM}_a(x)=a_{|x|}$. When $d>0$ it may be convenient to think of $\f{EXACT}_{k,l}^n$ in this way. In this representation $\f{EXACT}_{k,l}^n$ corresponds to lists $a$ of length $n+1$ with two 1s and the number of zeroes before the first, after the last 1, and distance between 1s correspond to parameters $k$, $n-l$, and $d$ respectively.

The boundary cases, $d=0$ and $d=n$, have been solved previously. When $d=n$, the function is usually referred to as $\f{EQUALITY}_n$. It can be solved with $n-1$ quantum queries using the well-known XOR trick. The case $d=0$ is also known as the $\f{EXACT}_k^n$ function which has been analyzed in \cite{AIS13} where it was shown that $Q_E(\f{EXACT}_k^n)=\max{\{k,n-k\}}$. In this paper, we completely solve the $d\in\{2,3\}$ cases and partially solve the $d=1$ case and $d\geq 4$ case.

The first of our results is
\begin{theorem}
  \label{thm:d1}
  If $d=1$, $l=n-k$ and $k>0$, then for $\f{EXACT}_{k,l}^n=\f{EXACT}_{k,k+1}^{2k+1}$
  \[Q_E(\f{EXACT}_{k,k+1}^{2k+1})=k+1.
\]
\end{theorem}
The algorithm we provide in the proof works also when $l \neq n-k$ by padding the function. However, the algorithm is then only an upper bound on $Q_E(\f{EXACT}_{k,k+1}^{n})$. For example,  $Q_E(\f{EXACT}_{2,3}^3)=2$ but our algorithm uses 3 queries for the padded version of the function (if we pad the input with two zeroes, we end up computing $\f{EXACT}_{2,3}^5$). Furthermore, the computations by Montanaro et al. \cite{MJM15} suggest that $Q_E(\f{EXACT}_{3,4}^5)=3$ and $Q_E(\f{EXACT}_{4,5}^6)=4$. There, unlike the $\f{EXACT}_{2,3}^3$ case, we don't know what the optimal algorithm looks like.

Next, we have a complete understanding of the $d\in\{2,3\}$ case,
\begin{theorem}
\label{thm:d23}
  If $d\in\{2,3\}$, then
  \[Q_E(\f{EXACT}_{k,l}^n)=\max\{n-k, l\}-1.\]
\end{theorem}
In particular, when $d=2$ and $l=n-k$, we have $l=k+2$ and $n = 2k+2$, meaning $l = \frac{n}{2}+1$, giving us $Q_E(\f{EXACT}_{k,l}^n)=\frac{n}{2}$ whereas the deterministic query complexity is $D(\f{EXACT}_{k,l}^n)=n$, hence we exhibit a factor of 2 gap between $Q_E(f)$ and $D(f)$ which is the largest known gap for a symmetric boolean function.

For larger values of $d$, we provide an exact quantum algorithm and a lower bound that is 2 queries less than the complexity of the algorithm:
\begin{theorem}
\label{thm:lowd4}
  If $d \geq 4$, then
  \[\max\{n-k, l\}-1 \leq Q_E(\f{EXACT}_{k,l}^n) \leq \max\{n-k, l\}+1,\]
\end{theorem}
We conjecture that our lower bound is tight, i.e., that 
\begin{conjecture}
  If $d\geq 2$, then
  \[Q_E(\f{EXACT}_{k,l}^n) = \max{\{n-k,l\}}-1.\]
\end{conjecture}
The lower bound of Theorem~\ref{thm:lowd4} combined with Theorem~\ref{thm:d1} implies that
\[Q_E(\f{EXACT}_{k,l}^n)\geq \frac{n}{2}.\]

Interestingly, the algorithm of Theorem \ref{thm:lowd4} can be used to compute a wide variety of symmetric functions with asymptotically optimal number of queries. Namely, we show
\begin{theorem}
\label{thm:sym}
Let $a\in \{0,1\}^{n+1}$ be a binary string with no 1-s far from its center, i.e. there exists some $g(n)\in o(n)$ such that $|i-\frac{n}{2}|>g(n)\implies a_i=0$. Then,
\[
   Q_E(\f{SYM}_a) = \frac{n}{2}+o(n).
\]
\end{theorem}
Since $D(\f{SYM}_a)=n$ for any such function $\f{SYM}_a$ (except for one that is 0 on all inputs), we obtain a factor-$(2-o(1))$ advantage for exact quantum algorithms for any
such $\f{SYM}_a$.

The outline for the rest of the paper is as follows. We describe the lower bound parts of Theorems \ref{thm:d1}, \ref{thm:d23} and \ref{thm:lowd4} in section \ref{sec:lb} and the algorithms for these theorems in section \ref{sec:thealg}. The algorithm for Theorem \ref{thm:sym} is given in Appendix \ref{sec:symalg}.

\section{The lower bounds}
\label{sec:lb}

\subsection{Proofs of the lower bound theorems}
\begin{theorem}
  \label{thm:lb}
  If $d\geq 1$, then
  \[Q_E(\f{EXACT}_{k,l}^n)\geq \max{\{n-k,l\}}-1.\]
\end{theorem}

This theorem provides the lower bound part for Theorems \ref{thm:d23} and \ref{thm:lowd4}.

\begin{proof}[Proof of Theorem \ref{thm:lb}]
  Consider the function $\f{EXACT}_{k,l}^n$ with $l\leq n-k$ ($l \geq n-k$ is symmetric and gives the $l-1$ result in the theorem). If the first $k$ input bits are ones, a quantum algorithm computing $\f{EXACT}_{k,l}^n$ must be computing $\f{EXACT}_{0,l-k}^{n-k}$ on the remaining $n-k$ input bits. Next we proceed similarly as in the lower bound via polynomials for $\f{OR}_n$ function\cite{BBC+98}. There must exist a state $\ket{\psi(x)}\in \H_Q \otimes \H_W \otimes \ket{1}$ which for $x=(0,\ldots,0)$ is non-zero at the end of the computation. If the algorithm performs $t$ queries, then the amplitude of the state $\ket{\psi(x)}$ can be expressed as a degree $\leq t$ multilinear polynomial in $\hat{x}$:
  \[p(\hat{x}_1,\ldots, \hat{x}_n)=\sum_{\substack{S:S\subseteq [n]\\ |S|\leq t}}{\alpha_S\prod_{i\in S}{\hat{x}_i}}.\]
  Let $p_{sym}$ be the symmetric polynomial
  \[p_{sym}(\hat{x}_1,\ldots, \hat{x}_n)=\sum_{\pi\in S_n}{\frac{p(\hat{x}_{\pi(1)},\ldots, \hat{x}_{\pi(n)})}{n!}}.\]
  Crucially, for the inputs $x\in \{(0,\ldots, 0)\}\cup \{x|\f{EXACT}_{0,l-k}^{n-k}(x)=0\}$:
  \[p_{sym}(\hat{x}_1,\ldots,\hat{x}_n)=p(\hat{x}_1,\ldots, \hat{x}_n).\]
  By assigning $s:=\frac{n-(\hat{x}_1+\ldots+\hat{x}_n)}{2}$ we can obtain a polynomial $q(s)$ that for all $\hat{x}\in \{-1,1\}^n$:
  \[q\left(\frac{n-(\hat{x}_1+\ldots+\hat{x}_n)}{2}\right)=p_{sym}(\hat{x}_1,\ldots, \hat{x}_n). \]
  The polynomial $q$ is therefore non-zero on $s=0$ and zero on $s\in \{0,1,\ldots, n-k\}\backslash \{0,l-k\}$. Thus it is a non-zero polynomial of degree at least $n-k-1$. On the other hand the degree of $q$ is at most $t$. Thus $n-k-1\leq \deg{q}\leq t$.
\end{proof}
This lower bound is not tight when $d=1$ and $l=n-k$. In this case we use a more sophisticated approach and give a different but possibly more insightful proof.

\begin{theorem}
  \label{thm:lbd1}
  If $d=1$, $n>1$ and $l=n-k$, then for $EXACT_{k,l}^n=EXACT_{k,k+1}^{2k+1}$
  \[Q_E(\f{EXACT}_{k,k+1}^{2k+1})\geq k+1.\]
\end{theorem}
Theorem~\ref{thm:lbd1} yields a lower bound that is better by one query than Theorem~\ref{thm:lb}, which yields a lower bound of $k$.

To show Theorem~\ref{thm:lbd1}, we use an unpublished result by Blekherman.
\begin{theorem}[Blekherman]
\label{thm:blek}
Let $q(\hat x)$ be the symmetrization of a polynomial $p^2(\hat x_1, \ldots, \hat x_n)$ 
where $p(\hat x)$ is a polynomial of degree $t\leq \frac{n}{2}$.
Then, over the Boolean hypercube $\hat{x}\in\{-1,1\}^n$,
\[ q(\hat x) = \sum_{j=0}^t p_{t-j}(|x|) \left( \prod_{0\leq i < j} (|x|-i) (n-|x|-i) \right) \]
where $p_{t-j}$ is a univariate polynomial that is a sum of squares of polynomials of degree at most $t-j$
and $|x|$ denotes the number of variables $i:\hat x_i=-1$.
\end{theorem}
See \cite{LP+16} for a proof of Blekherman's theorem. Furthermore, we provide a considerably shorter proof in the next subsection.

\begin{proof}[Proof of Theorem \ref{thm:lbd1}]
  Let us consider the negation of the function $\f{EXACT}_{k,k+1}^{2k+1}$. Assuming, towards a contradiction, that there exists a quantum algorithm computing the function with $k$ queries, there exists a sum of squares representation of $\f{NOT-EXACT}_{k,k+1}^{2k+1}$:
  \[\f{NOT-EXACT}_{k,k+1}^{2k+1}(x)=\sum_i{r_i^2(\hat{x})},\]
  such that $\deg{r_i}\leq k$. Since the function is symmetric, the symmetrization is also a valid representation. Since $Sym(\sum_i{r_i^2(\hat{x})})=\sum_i{Sym(r_i^2(\hat{x}))}$, it follows from Blekherman's theorem that there is a univariate polynomial of the form
  \begin{equation}
    \label{eq:symsos}
  q(|x|)=\sum_{j=0}^k {p_{k-j}(|x|) \left( \prod_{i=0}^{j-1} (|x|-i) (n-|x|-i) \right)},
  \end{equation}
  where $q(|x|)=\f{NOT-EXACT}_{k,k+1}^{2k+1}(x)$ on the Boolean hypercube and $p_{k-j}$ are sum of squares polynomials with $\deg{p_{k-j}}\leq 2k-2j$. The polynomial $q(|x|)$ is non-negative in the interval $|x|\in [k-1,k+2]$. Since the polynomial is 0 at $|x|=k$ and $|x|=k+1$, it must have at least 3 local extrema in the interval $|x|\in[k,k+1]$. Additionally, it is 1 when $|x|\in\{0,1,\ldots,n\}\backslash\{k,k+1\}$, hence it has $2k-2$ more extrema. In total the polynomial has at least $2k+1$ local extrema, therefore its degree is at least $2k+2$. On the other hand by our assumption $\deg{q}\leq 2k$ which is a contradiction.
\end{proof}


\subsection{Proof of Blekherman's theorem}

\subsubsection{Group representation}

Let $H_\wp$ be a Hilbert space with basis states $\hat x_{S}$ (for all $S\subseteq [n]$) corresponding to monomials $\prod_{i\in S}\hat  x_i$. Then, the vectors in $H_\wp$ correspond to multilinear polynomials in variables $\hat x_i$.
We consider a group representation of the symmetric group $\mathfrak{S}_n$ on $H_\wp$ with transformations $U_{\pi}$ defined by 
$U_{\pi}\hat  x_{S}=\hat x_{\pi(S)}$.
The irreducible representations contained in $H_\wp$ are well known:

Let $S_m(\hat{x}_1, \ldots, \hat{x}_n)=\sum_{i_1, \ldots, i_m} \hat{x}_{i_1} \ldots \hat{x}_{i_m}$ be the $m^{\rm th}$ elementary symmetric polynomial. We use $S_0(\hat{x}_1, \ldots, \hat{x}_n)$ to
denote the constant 1.

\begin{lemma}
  \label{lem:irred}
 A subspace $H\subseteq H_\wp$ is irreducible if and only if there exist $b$ and $\alpha_{m}$ for $m= 0, 1, \ldots, n-2b$ such that
  $H$ is spanned by vectors $\overrightarrow{p}_{i_1, \ldots, j_b}$  corresponding to polynomials
$p_{i_1, \ldots, j_b}$  (for all choices of pairwise distinct $i_1, j_1, \ldots, i_b, j_b\in [n]$) where
  \[ p_{i_1, \ldots, j_b}(\hat{x}_1, \ldots, \hat{x}_n) = (\hat{x}_{i_1}-\hat{x}_{j_1}) \ldots (\hat{x}_{i_b}-\hat{x}_{j_b}) \sum_{m=0}^{n-2b} \alpha_m S_m(\hat{x}') \]
and $\hat{x}'\in\{-1, 1\}^{n-2b}$ consists of all $\hat{x}_i$ for $i\in [n]$, $i\notin\{i_1, \ldots, j_b\}$.
\end{lemma}

See \cite{Bel15} for a short proof of Lemma \ref{lem:irred}.

\subsubsection{Decomposition of $q(\hat x)$}

Let
\[ p(\hat x_1, \ldots,\hat  x_n) = \sum_{S: |S|\leq t} a_S \hat x_S .\]
We associate $p^2(\hat x_1, \ldots, \hat x_n)$ with the matrix $(P_{S_1, S_2})$ 
with rows and columns indexed by $S\subseteq [n], |S|\leq t$ defined by $P_{S_1, S_2}=a_{S_1} a_{S_2}$.
Let $\overrightarrow{x}$ be a column vector consisting of all $\hat x_S$ for $S:|S|\leq t$. 
Then, $p^2(\hat x_1, \ldots, \hat x_n) = \overrightarrow{x}^T P \overrightarrow{x}$.
This means that $P$ is positive semidefinite.

For a permutation $\pi\in \mathfrak{S}_n$, let $P^{\pi}$ be the matrix defined by 
\[ P^{\pi}_{S_1, S_2} = a_{\pi(S_1)} a_{\pi(S_2)} \]
and let $Q=\frac{1}{n!} \sum_{\pi\in \mathfrak{S}_n} P^{\pi}$ be the average of all $P^{\pi}$. Then, 
$q(\hat x)=\overrightarrow{x}^T Q \overrightarrow{x}$.
$Q$ is also positive semidefinite (as a linear combination of positive semidefinite matrices $P^{\pi}$ with positive coefficients).

We decompose $Q=\sum_i \lambda_i Q_i$ with $\lambda_i$ ranging over different non-zero eigenvalues and $Q_i$ being the projectors on the respective eigenspaces. Since $Q$ is positive semidefinite, we have $\lambda_i>0$ for all $i$.

We interpret transformations $U_{\pi}$ as permutation matrices defined by $(U_\pi)_{S, S'}=1$ if $S=\pi(S')$ and 
$(U_\pi)_{S, S'}=0$ otherwise.
Then, we have 
\[ U_{\pi} Q U^{\dagger}_{\pi} = \frac{1}{n!} \sum_{\tau\in \mathfrak{S}_n} U_{\pi} P^{\tau} U^{\dagger}_{\pi} =
\frac{1}{n!} \sum_{\tau\in \mathfrak{S}_n} P^{\pi \tau} = \frac{1}{n!} \sum_{\tau\in \mathfrak{S}_n} P^{\tau} = Q.\]
Since we also have
\[ U_{\pi} Q U^{\dagger}_{\pi} = \sum_i \lambda_i U_{\pi} Q_i U^{\dagger}_{\pi}, \]
we must have $Q_i= U_{\pi} Q_i U^{\dagger}_{\pi}$. This means that $Q_i$ is a projector to 
a subspace $H_i\subseteq H_\wp$ that is invariant under the action of $\mathfrak{S}_n$.
If $H_i$ is not irreducible, we can decompose it into a direct sum of irreducible subspaces
\[ H_i = H_{i, 1} \oplus H_{i, 2} \oplus \ldots \oplus H_{i, m_i} .\]
Then, we have $Q_i = \sum_{j=1}^{m_i} Q_{i, j}$ where $Q_{i, j}$ is a projector to $H_{i, j}$
and $Q=\sum_{i, j} \lambda_i Q_{i, j}$.
This means that we can decompose $q(\hat x)=\sum_{i, j}\lambda_i q_{i, j}(\hat x)$
where $q_{i, j}(\hat x)=\overrightarrow{x}^T Q_{i, j} \overrightarrow{x}$ 
and it suffices to show the theorem for one polynomial $q_{i, j}(\hat x)$ instead of
the whole sum $q(\hat x)$.

\subsubsection{Projector to one subspace.}
Let $H_{\wp,\ell}\subseteq H_{\wp}$ be an irreducible invariant subspace.
We claim that the projection to the subspace $H_{\wp,\ell}$ denoted by $\Pi_{\wp,\ell}$ is of the following form:
\begin{lemma}
  \[ \Pi_{\wp, \ell} = c\rho_{\wp, \ell} \text{ where } 
  \rho_{\wp, \ell} =\sum_{i_1, \ldots, j_b} \overrightarrow{p}_{i_1, \ldots, j_b} \overrightarrow{p}^T_{i_1, \ldots, j_b} 
  \]
for some constant $c$.
\end{lemma}

\proof
If we restrict to the subspace $H_{\wp, \ell}$, then $\Pi_{\wp, \ell}$ is just the identity $I$.

On the right hand side, $\rho_{\wp, \ell}$ is mapped to itself by any $U_{\pi}$ (since any $U_{\pi}$ permutes the vectors
$\overrightarrow{p}_{i_1, \ldots, j_b}$ in some way). 
Therefore, all $U_{\pi}$ also map the eigenspaces of $\rho_{\wp, \ell}$ to themselves.
This means that, if $\rho_{\wp, \ell}$ has an eigenspace $V\subset H_{\wp, \ell}$, 
then $U_{\pi}$ acting on $V$ also form a representation of $\mathfrak{S}_n$ but that would contradict
$H_{\wp, \ell}$ being an irreducible representation.
Therefore, the only eigenspace of $\rho_{\wp, \ell}$ is the entire $H_{\wp, \ell}$.
This can only happen if $\rho_{\wp, \ell}$ is $c I$ for some constant $c$.
\qed

\subsubsection{Final polynomial}

From the previous subsection, it follows that $q_{i, j}(\hat x)$ is a positive constant times
\[ \sum_{i_1, \ldots, j_b}  (\hat{x}_{i_1}-\hat{x}_{j_1})^2 \ldots (\hat{x}_{i_b}-\hat{x}_{j_b})^2 S^2(\hat{x}') \]
 where $S(\hat{x}')$ is a symmetric polynomial of degree at most $t-b$.
Instead of the sum, we consider the expected value of 
$(\hat{x}_{i_1}-\hat{x}_{j_1})^2 \ldots (\hat{x}_{i_b}-\hat{x}_{j_b})^2 S^2(\hat{x}')$ 
when $i_1, \ldots, j_b$ are chosen randomly. (Since the sum and the expected value differ by a constant factor,
this is sufficient.)

Terms $(\hat x_{i_k} - \hat x_{j_k})^2$ are nonzero if and only if one of $x_{i_k}$ and $x_{j_k}$ is $1$ and the other is $-1$.
Then, for $k=1$, we have
\[ Pr\left[\{\hat x_{i_1}, \hat x_{j_1}\} =\{-1, 1\} \right] = \frac{2s(n-s)}{n(n-1)} ,\]
since there are $\frac{n(n-1)}{2}$ possible sets $\{\hat x_{i_1}, \hat x_{j_1}\}$ and $s(n-s)$ of 
them contain one $1$ and one $-1$.
For $k>1$,
\[ Pr\left[\{\hat x_{i_k}, \hat x_{j_k}\} =\{-1, 1\} |
\{\hat x_{i_l}, \hat x_{j_l}\} =\{-1, 1\} \mbox{ for } l\in[k-1] \right] \]
\[ = \frac{2(s-k+1)(n-s-k+1)}{(n-2k+2)(n-2k+1)} ,\]
since the condition $\{\hat x_{i_l}, \hat x_{j_l}\} =\{-1, 1\}$ for $l\in [k-1]$ means that, among the remaining variables, there are
$s-k+1$ variables $\hat x_j=-1$ and $n-s-k+1$  variables $\hat x_j=1$ and $n-2k+2$ variables in total (and, given that,
the $k=1$ argument applies). Thus,
\[ Pr [ (\hat{x}_{i_1}-\hat{x}_{j_1})^2 \ldots (\hat{x}_{i_b}-\hat{x}_{j_b})^2 =1 ] = \frac{2^b s(s-1) \ldots (s-b+1) (n-s) \ldots (n-s-b+1)}{n(n-1) \ldots (n-2b+1)} .\]
Since $S$ is a symmetric polynomial, we have $S(\hat{x}')=S'(s')$ where $S'$ is a polynomial of one variable $s'$, with $s'$ equal to
the number of variables $\hat{x}'_j=-1$. Since there are $b$ variables $\hat{x}_j=-1$ that do not appear in $\hat{x}'$,
we have $s'=s-b$. This means that $S'$ can be rewritten as a polynomial in $s$ (instead of $s'$).


\section{The algorithms}
\label{sec:thealg}

In Section~\ref{sec:unb3} we now provide the algorithm for $d\leq 3$
(the algorithm part of Theorems \ref{thm:d1} and \ref{thm:d23}) which we know to be optimal for $d=1$ with $k+l=n$, and for $d=2,3$ and any $k,l$.
Next, in Section~\ref{sec:generalalg} we present the sub-optimal algorithm that works for all $d$, resulting in a general upper bound  on $Q_E(\f{EXACT}_{k,l}^{n})$ (the algorithm part of Theorem \ref{thm:lowd4}). Throughout Section~\ref{sec:thealg} we will refer to $\hat{x}_1+\ldots+\hat{x}_n$ as the unbalance of the input or simply unbalance, in other words, the unbalance increases as the difference between ones and zeroes in the input increases. When $k+l=n$, the condition $\f{EXACT}_{k,n-k}^n(x)=1$ is equivalent to requirement that the unbalance is $\pm d$, i.e., $|\hat{x}_1+\ldots+\hat{x}_n|=n-2k=d$. Hence we will refer to $\f{EXACT}_{k,n-k}^n$ as testing for unbalance $d=n-2k$.

\subsection{The algorithm for unbalance $d\leq 3$}
\label{sec:unb3}
For the upper bound, we now provide a quantum algorithm for the $l=n-k$ case which can be extended to $l\neq n-k$ case. Let us introduce the function $\f{UNBALANCE}_d^n=\f{EXACT}_{\frac{n-d}{2},\frac{n+d}{2}}^n$. When $l=n-k$ then $d=n-2k$ and so $n$ and $d$ have the same parity.

\begin{theorem}
  \label{thm:alg}
  \[Q_E(\f{UNBALANCE}_d^n)\leq \begin{cases}\frac{n+d}{2}\text{, if }d=1,\\
      \frac{n+d}{2}-1\text{, if }d\in\{2,3\}.
  \end{cases}
  \]
\end{theorem}

We can compute $\f{EXACT}_{k,l}^n$ for $l\neq n-k$ by reducing it to $\f{UNBALANCE}_{d'}^{n'}$:
\begin{lemma}
    \label{lem:red}
    \[Q_E(\f{EXACT}_{k,l}^n) \leq Q_E\left(\f{UNBALANCE}_{l-k}^{n+\max{\{n-l-k,l+k-n\}}}\right)\]
\end{lemma}
\begin{proof}
For the $l<n-k$ case ($l>n-k$, respectively) simply pad the input bits with $n-l-k$ ones ($l+k-n$ zeroes, resp.) and run $\f{UNBALANCE}_d^{n+|n-l-k|}$ on the padded input. The complexity of the algorithm on the padded problem will be
\[Q_E(\f{EXACT}_{k,l}^n)\leq Q_E\left(\f{UNBALANCE}_d^{n+|n-l-k|}\right).\]
\end{proof}

From Lemma~\ref{lem:red} and Theorem~\ref{thm:alg}, the upper bounds of Theorem~\ref{thm:d1} and Theorem~\ref{thm:d23} follow: 
\[Q_E(\f{EXACT}_{k,l}^n)\leq\begin{cases}
\max{\{n-k,l\}}\text{, if }d=1,\\
\max{\{n-k,l\}}-1\text{, if }d\in \{2,3\}.
\end{cases}\]


\subsubsection{The structure of the algorithm}

  The algorithm of Theorem~\ref{thm:alg} will use two kinds of subroutines to calculate the function:
  \begin{itemize}
  \item The main routine $\sr{UNB}_d^n$ will start in a quantum state independent of the input and compute a $\f{UNBALANCE}_d^n$ instance;
  \item The subroutine $\sr{UNB-R}_d^n$ will require a precomputed state in the form
  \begin{equation}
  \label{eq:ss}
  \sum_{i\in[n]}{\hat{x}_i}\ket{\S}+\sqrt{\gamma}\sum_{\substack{i,j\in[n]\\i<j}}{(\hat{x}_i-\hat{x}_j)\ket{i,j}}.
  \end{equation}
  $\ket{\S}$ here alludes to fact that the amplitude of the basis state is a sum of $\hat{x}_i$'s.
  \end{itemize}
  
  Let us denote by $\gamma(\sr{UNB-R}_d^n)$ the constant coefficient $\gamma$ of the algorithm $\sr{UNB-R}_d^n$.

    Let us denote by $T(\sr{S})$ the number of queries performed by a subroutine $\sr{S}$.

\begin{lemma}[Recursive step for $\sr{UNB-R}_d^n$]
  \label{lem:rs}
  If $d < n$, $n \geq 3$, and there exists a quantum algorithm $\sr{UNB-R}_d^{n-2}$ computing the function $\f{UNBALANCE}_d^{n-2}$ starting in an unnormalized quantum state of the form (\ref{eq:ss}) on $n-2$ inputs with $\gamma(\sr{UNB-R}_d^{n-2})<1$ then there exists an algorithm $\sr{UNB-R}_d^n$ using $\sr{UNB-R}_d^{n-2}$ as a subroutine, and computing $\f{UNBALANCE}_d^{n}$, starting in the state (\ref{eq:ss})
  where
  \begin{equation}
  \label{eq:gamma}
  \gamma(\sr{UNB-R}_d^n)
  =
  \frac{1}{(n^2-d^2)^2}
  \left(
     n^2(n-2)^2\frac{\gamma(\sr{UNB-R}_d^{n-2})}{1-\gamma(\sr{UNB-R}_d^{n-2})}+d^4
  \right)
  \end{equation}
  and using one more query, i.e.,
\[T(\sr{UNB-R}_d^n)=T(\sr{UNB-R}_d^{n-2})+1.\]
\end{lemma}

The main routine $\sr{UNB}_d^n$ will also be recursive and make use of $\sr{UNB-R}_d^n$. 

\begin{lemma}[Recursive step for $\sr{UNB}_d^n$]
\label{lem:mainrs}
If there exists $\sr{UNB}_d^{n-2}$ and $\sr{UNB-R}_d^n$ with $\gamma(\sr{UNB-R}_d^n)\leq 1$, then there exists $\sr{UNB}_d^n$ such that
\[T(\sr{UNB}_d^n)=1+\max\{T(\sr{UNB}_d^{n-2}),T(\sr{UNB-R}_d^n)\}.\]
\end{lemma}

Now we are ready to prove Theorem \ref{thm:alg}:
\begin{proof}[Proof of Theorem \ref{thm:alg}]
  We can draw the subroutine dependency graph like so:
\[\begin{array}{ccccccccc}
  \sr{UNB}_d^d & \leftarrow & \sr{UNB}_d^{d+2} &\leftarrow & \sr{UNB}_d^{d+4} & \leftarrow & \cdots & \leftarrow & \sr{UNB}_d^{d+2k}\\
  & & \downarrow & & \downarrow & & & & \downarrow \\
  \sr{UNB-R}_d^d & \leftarrow & \sr{UNB-R}_d^{d+2} & \leftarrow & \sr{UNB-R}_d^{d+4} & \leftarrow & \cdots & \leftarrow & \sr{UNB-R}_d^{d+2k}
\end{array}\]
  Each subroutine performs one query and calls one of the subroutines in the dependency graph depending on the result of the measurement. Using Lemma \ref{lem:rs} starting with an algorithm $\sr{UNB-R}_d^{d+2k_0}$ we can build chains of algorithms $\sr{UNB-R}_d^{d+2k_0},\sr{UNB-R}_d^{d+2(k_0+1)},\ldots,\sr{UNB-R}_d^{d+2k}$ as long as $\gamma(\sr{UNB-R}_d^{d+2k_i})<1$.
Notice that we may use multiple chains to cover all $k>0$. Fortunately, as we will show for $d\in\{1,2,3\}$, a single infinite chain will suffice.

Then, using Lemma \ref{lem:mainrs} we can build algorithms $\sr{UNB}_d^{d+2k}$ for all $k>0$ if we additionally have an initial base algorithm for $\sr{UNB}_d^d$. The query complexity of $\sr{UNB}_d^{d+2k}$ built in this way on a chain of $\sr{UNB-R}_d^{d+2k}$ starting at $k_0\in\{0,1\}$ will have
\[T(\sr{UNB}_d^{d+2k}) = \max\{k+T(\sr{UNB}_d^d),T(\sr{UNB-R}_d^{d+2k_0})+k-k_0+1\}.\]

Since $\sr{UNB}_d^d$ is computing $\f{EQUALITY}_d$, it uses $d-1$ queries, so we can disregard $k+T(\sr{UNB}_d^d)$, since $k=\frac{n-d}{2}$ and therefore $k+T(\sr{UNB}_d^d)\leq \frac{n+d}{2}-1$. To finish the proof we now need to show that there exists a chain of $\sr{UNB-R}_d^{d+2k}$ starting at $k_0$ with $\gamma(\sr{UNB-R}_d^n)<1$ and
\[T(\sr{UNB-R}_d^{d+2k_0})+k-k_0+1\leq\begin{cases}n-k\text{, if }d=1,\\
n-k-1\text{, if }d\in\{2,3\}.
\end{cases}
\]

\begin{itemize}
    \item When $d=1$, we will have $k_0=0$ and show that $T(\sr{UNB-R}_d^d)\leq n-2k+k_0-1=d+k_0-1=0$.
Since the function $\f{UNBALANCE}_1^1$ does not depend on input variables, there exists $\sr{UNB-R}_1^1$ with $\gamma(\sr{UNB-R}_1^1)=0$ using 0 queries.

    \item When $d=2$ we will again have $k_0=0$ and $T(\sr{UNB-R}_d^d)\leq d+k_0-2=0$.
The subroutine $\sr{UNB-R}_2^2$ is essentially required to compute $XOR(x_1,x_2)$ starting in a non-normalized state $(\hat{x}_1+\hat{x}_2)\ket{\S}+\sqrt{\gamma}\cdot(\hat{x}_1-\hat{x}_2)\ket{1,2}$. If $\gamma=0$ we can only measure $\ket{\S}$ if $XOR=0$ and no queries are necessary.

    \item When $d=3$ a single infinite chain starting at $k_0=0$ does not exist. It does exist starting at $k_0=1$ and $T(\sr{UNB-R}_d^{d+2})\leq d+k_0-2=2$.
We give algorithm for this as a separate lemma:
\begin{lemma}
\label{lem:1or4}
There exists a subroutine $\sr{UNB-R}_3^5$ with $\gamma(\sr{UNB-R}_3^5)=\frac{1}{112}$ using 2 queries.
\end{lemma}
\end{itemize}

  To show that the chains of algorithms $\sr{UNB-R}_d^{d+2k}$ obtained by repeated applications of Lemma \ref{lem:rs} never have $\gamma(\sr{UNB-R}_d^{d+2k})\geq 1$, we use the recursive identity (\ref{eq:gamma}). It would be sufficient to show that $\exists n_{init} \forall n \geq n_{init}:\gamma(\sr{UNB-R}_d^n)\leq \frac{1}{n}$. For $n < n_{init}$ it can be verified through explicit computation. For this it would be sufficient to show that $\exists n_{init} :\gamma(\sr{UNB-R}_d^{n_{init}})\leq \frac{1}{n_{init}}\wedge \forall n > n_{init}:\gamma(\sr{UNB-R}_d^{n-2})\leq \frac{1}{n-2} \rightarrow \gamma(\sr{UNB-R}_d^n)\leq \frac{1}{n}$. The implication holds whenever
  \[\frac{\frac{n^2(n-2)^2}{n-3}+d^4}{(n^2-d^2)^2}\leq \frac{1}{n},\]
or equivalently,
  \[n^4+(-2d^2-4)n^3+(6d^2-d^4)n^2+4d^4n-3d^4\geq 0.\]
  When $d=1$ the inequality holds for $n\geq 5$. We can then numerically verify that $\gamma(\sr{UNB-R}_1^5) \approx 0.008 \leq \frac{1}{5}$. When $d=2$ the inequality holds onwards from $n\geq 12$. For our base case $\gamma(\sr{UNB-R}_2^{12})\approx  0.039\leq \frac{1}{12}$. When $d=3$ the inequality holds onwards from $n \geq 23$. For our chain $\gamma(\sr{UNB-R}_3^{23})\approx 0.030 \leq \frac{1}{23}$.

\end{proof}


\subsection{Proof of Lemma \ref{lem:rs}}
\begin{proof}
  Our algorithm will utilize the following two unitaries and their inverses:
  \begin{itemize}
  \item 
    $R_\alpha$ is a unitary transformation over a 3-dimensional Hilbert space with basis vectors $\ket{0},\ket{\L}, $ and $\ket{\R}$. 
    It is a unitary completion of the following transformation:
    \[R_\alpha\ket{0} = \sin{\alpha}\ket{\L} + \cos{\alpha}\ket{\R}.\]
  \item
    $U_n$ is a unitary transformation over a Hilbert space of dimension $n+\binom{n}{2}+1$ with basis vectors $\{\ket{1}, \ket{2}, \ldots, \ket{n}, \ket{\S}, \ket{1,2}, \ket{1,3}, \ldots, \ket{n-1,n}\}$. It is a unitary completion of the following transformation:
\begin{align}
    U_n\ket{i} = \frac{1}{\sqrt{n}}
        \Bigg(
        \ket{\S}
        - \sum_{j=1}^{i-1}
            \ket{j,i}
            + \sum_{j=i+1}^{n} 
            \ket{i,j}
        \Bigg).
    \label{Uoperation}
\end{align}
    Note that on a superposition of input vectors $U_n$ acts as:
\[
    U_n \sum_{i\in[n]}
        \alpha_i\ket{i}
    = \frac{1}{\sqrt{n}}
        \Bigg(
            \sum_{i\in[n]}{\alpha_i\ket{\S}}+\sum_{\substack{i,j\in[n]\\i<j}}{(\alpha_i-\alpha_j)\ket{i,j}}
        \Bigg).
\]
  \end{itemize}
  The subspace $\{\ket{1},\ldots,\ket{n}\}$ can be regarded as the input subspace of $U_n$ and the orthogonal subspace $\{\ket{\S},\ket{1,2},\ket{1,3},\ldots,\ket{n-1,n}\}$ --- as the output subspace. We will call $\ket{\S}$ the sum output state and $\{\ket{1,2},\ket{1,3},\ldots,\ket{n-1,n}\}$ the difference output states. In the description of the algorithm we will specify which basis states are designated as input and output states for each $R_\alpha$ and $U_n$.
  
  We will track the state of the algorithm $\sr{UNB-R}_d^n$  throughout the recursive step. Additionally, we will introduce some real constants and specify the constraints on them induced by the algorithm. Let $\gamma=\gamma(\sr{UNB-R}_d^n)$ and $\gamma'=\gamma(\sr{UNB-R}_d^{n-2})$. The algorithm starts in the state:
  \[\sum_{i\in[n]}{\hat{x}_i}\ket{\S}+\sqrt{\gamma}\sum_{\substack{i,j\in[n]\\i<j}}{(\hat{x}_i-\hat{x}_j)\ket{i,j}}.\]
  We now apply $R_\alpha$ to each of the $\ket{i,j}$ and obtain
  \[\sum_{i\in[n]}{\hat{x}_i}\ket{\S}+c_1\sum_{\substack{i,j\in[n]\\i<j}}{(\hat{x}_i-\hat{x}_j)\ket{i,j,\L}}+c_2\sum_{\substack{i,j\in[n]\\i<j}}{(\hat{x}_i-\hat{x}_j)\ket{i,j,\R}},\]
  with
  \begin{equation*}
    \tag{C1}
    c_1^2+c_2^2=\gamma.
  \end{equation*}
  Let us apply $U_n^{-1}$ to the $\ket{\S}$ and $\ket{i,j,\L}$ parts of the state with $\ket{\S}$ and $\ket{i,j,\L}$ serving as the input states of $U_n^{-1}$. The output states for $U_n$ are $\{\ket{l}|l\in [n]\}$. For each state $\ket{i,j,\R}$ we perform $U_{n-2}^{-1}$ with $\ket{i,j,\R}$ serving as the sum input state of $U_{n-2}^{-1}$ and $\ket{i,j,\R,u,v}$ being some auxiliary input states with 0 amplitudes corresponding to difference input states $\{\ket{u,v}| \{u,v\}\subseteq[n]\backslash\{i,j\}\}$). The output states for $\ket{i,j}$ are $\ket{i,j,l},l\in[n]\backslash \{i,j\}$. We obtain:
\[
    c_3\sum_{l\in[n]}{\bigg(\sum_{i\in[n]}\hat{x}_i-c_4\hat{x}_l\bigg)\ket{l}}
    +c_5\sum_{\substack{i,j\in[n]\\l\in[n]\backslash\{i,j\}\\i<j}}{(\hat{x}_i-\hat{x}_j)\ket{i,j,l}}.
\]
It is easier to verify this statement by working backwards --- pretending that we apply $U_n$ and $U_{n-2}$, respectively, to the state above. Unlike their inverses, we know how to apply $U_n$ and $U_{n-2}$. Combined with the following constraints, the above can be verified.
\begin{align*}
    c_3\cdot n-c_3\cdot c_4=\sqrt{n}, \qquad \qquad
    c_3\cdot c_4=-c_1\sqrt{n}, \qquad \qquad
    c_2=c_5\sqrt{n-2}.
    \tag{C2-C4}
\end{align*}
  The constraints can be obtained by considering the coefficient of the terms before and after the transformation. For example, the first constraint (C2) is the coefficient in front of $\sum_{i\in[n]}{\hat{x}_i}$ before the transformation. If we run the algorithm backwards, the coefficient only depends on $c_3$ and $c_4$ from the states $c_3\sum_{l\in[n]}{\big(\sum_{i\in[n]}\hat{x}_i-c_4\hat{x}_l\big)\ket{l}}$.
  
  Next, we query the variable as specified by the last number in the register, getting
\[
    c_3\sum_{l\in[n]}{\hat{x}_l
        \bigg(
            \sum_{i\in[n]}\hat{x}_i-c_4\hat{x}_l
        \bigg)\ket{l}}
    +c_5\sum_{\substack{i,j\in[n]\\l\in[n]\backslash\{i,j\}\\i<j}}{\hat{x}_l(\hat{x}_i-\hat{x}_j)\ket{i,j,l}}.\]
Next, we apply $U_n$ to the $\ket{l}$ states as input states and using $\ket{S}$ and $\ket{i,j,\L}$ as the output states. Next, for each pair $\{i,j\}$, we apply $U_{n-2}$ to the group of states $\{\ket{i,j,l}|l\in[n]\backslash\{i,j\}\}$, thinking of those as input states and $\ket{i,j,\R}$ playing the role of the sum output state $\ket{\S}$ and $\ket{i,j,u,v}$ having the role of difference output states. We obtain
\[
    \begin{split}
    c_6\bigg(\sum_{\substack{i,j\in[n]}}{\hat{x}_i\hat{x}_j}-c_7\bigg)\ket{\S}
    +c_8\sum_{\substack{i,j\in[n]\\l\in[n]\backslash\{i,j\}\\i<j}}{(\hat{x}_i-\hat{x}_j)\hat{x}_l\ket{i,j,\L}}\\
    +c_9\sum_{\substack{i,j\in[n]\\l\in[n]\backslash\{i,j\}\\i<j}}{(\hat{x}_i-\hat{x}_j)\hat{x}_l\ket{i,j,\R}}
    +c_{10}\sum_{\substack{i,j\in[n]\\u,v\in[n]\backslash\{i,j\}\\i<j\\u<v}}{(\hat{x}_i-\hat{x}_j)(\hat{x}_u-\hat{x}_v)\ket{i,j,u,v}},
  \end{split}
  \]
when the following constraints hold: 
\begin{align*}
   &c_3=c_6\sqrt{n}, \qquad \qquad 
        c_5=c_9\sqrt{n-2}, \qquad \qquad 
        c_3\cdot c_4\cdot n=c_6\cdot c_7\sqrt{n}, \\
   &c_3=c_8 \sqrt{n}, \qquad \qquad 
       c_5=c_{10}\sqrt{n-2}. \tag{C5-C9}
\end{align*}

  To finish up the unitary transformations of the recursive step we now perform $R_\alpha^{-1}$ on the pairs of states $\ket{i,j,\L}$ and $\ket{i,j,\R}$ states, turning them to $\ket{i,j}$, and giving us
  \[\begin{split}c_6\bigg(\sum_{i,j\in[n]}{\hat{x}_i\hat{x}_j}-c_7\bigg)\ket{\S}+\\
  +c_{11}\sum_{\substack{i,j\in[n]\\i<j}}{(\hat{x}_i-\hat{x}_j)\ket{i,j}\bigg(\sum_{\substack{l\in[n]\\l\notin\{i,j\}}}{\hat{x}_l\ket{\S}}+\sqrt{\gamma'}\sum_{\substack{u,v\in[n]\\u,v\notin\{i,j\}\\u<v}}{(\hat{x}_u-\hat{x}_v)\ket{u,v}}\bigg)}
  \end{split}
  \]
  Again, this is true if the constraints obey 
\[
    c_8^2+c_9^2=c_{11}^2, \qquad \qquad
    c_{10}=c_{11}\cdot \sqrt{\gamma'}. \tag{C10-C11}
\]
  Finally, we measure whether the state is in subspace $\{\ket{\S}\}$. We can set the constant $c_7$ so that whenever $\f{UNBALANCE}_d^n(x)=1$ or equivalently $\hat{x}_1+\ldots +\hat{x}_n=\pm d$ the amplitude of $\ket{\S}$ is zero:
  \begin{equation*}
    \tag{C12}
    c_7=d^2.
  \end{equation*}
  If on the other hand the state is not in subspace $\ket{\S}$, we end up measuring $\ket{i,j}$ in the first register. Without loss of generality we may assume that the result is $\{n-1,n\}$. Thus we have learned that $\{\hat{x}_{n-1},\hat{x}_n\}=\{-1,1\}$ is a balanced pair that can be removed from consideration. Furthermore, we ended up in a useful (unnormalized) state
  \[\sum_{i\in[n-2]}{\hat{x}_i}\ket{\S}+\sqrt{\gamma'}\sum_{\substack{i,j\in[n-2]\\i<j}}{(\hat{x}_i-\hat{x}_j)\ket{i,j}}.\]
  Therefore, we can call $\sr{UNB-R}_d^{n-2}$ recursively,  since
  \[\f{EXACT}^n_{k,n-k}(x_1, \ldots, x_n) = \f{EXACT}^{n-2}_{k-1,n-k-1}(x_1, \ldots, x_{n-2}).\]

  When we solve for $\gamma$ in terms of $n$, $d$ and $\gamma'$, there is only one solution up to the signs of some constants $c_i$. The solution is specified in the statement of Lemma~\ref{lem:rs}.
\end{proof}

\subsection{Proof of Lemma \ref{lem:mainrs}}
\begin{proof}
  We start in state $\sum_{i\in[n]}{\frac{1}{\sqrt{n}}\ket{i}}$, perform the query to get $\sum_{i\in[n]}{\frac{\hat{x}_i}{\sqrt{n}}\ket{i}}$ and apply $U_n$ from Lemma \ref{lem:rs} obtaining 
  \[\frac{1}{n}
  \Bigg(
      \sum_{i\in[n]}
  \hat{x}_i \ket{\S}
  +\sum_{\substack{i,j\in[n]\\i<j}}
      (\hat{x}_i-\hat{x}_j)\ket{i,j}
  \Bigg).\]
  Let $\gamma = \gamma(\sr{UNB-R}_d^n)$. Next, we apply $R_\gamma$ on the second part of the state, obtaining
  \[
  \begin{split}
  \frac{1}{n}
     \Bigg(
      \sum_{i\in[n]}{\hat{x}_i}\ket{\S}+\sqrt{\gamma}\sum_{\substack{i,j\in[n]\\i<j}}{(\hat{x}_i-\hat{x}_j)\ket{i,j,\L}}
      \Bigg)+\frac{1}{n}\sqrt{1-\gamma}\sum_{\substack{i,j\in[n]\\i<j}}{(\hat{x}_i-\hat{x}_j)\ket{i,j,\R}}
  \end{split}
  .\]
  Finally we measure completely the subspace labeled with $\R$. If we obtain $\ket{i,j,\R}$ we learn that $\hat{x}_i\neq \hat{x}_j$ and thus have reduced our problem to $\f{UNBALANCE}_d^{n-2}$ on the remaining variables which we can compute using $\sr{UNB}_{d,n-2}$. If we obtain the orthogonal subspace, we end up in non-normalized state
  \[
  \sum_{i\in[n]}{\hat{x}_i}\ket{\S}+\sqrt{\gamma}\sum_{\substack{i,j\in[n]\\i<j}}{(\hat{x}_i-\hat{x}_j)\ket{i,j,\L}},
  \]
  which we pass to $\sr{UNB-R}_d^n$.
\end{proof}


\subsection{The general upper bound}
\label{sec:generalalg}

We now present a general upper bound to
$Q(\f{EXACT}_{k,l}^{n})$. The algorithms we present are worse (by at most 2 queries) than the one in Section \ref{sec:thealg} when $l-k = d \leq 3$. However, they work for any $k,l$ and thus also for any $d$. 

First, for the special case $k+l=n$, we claim
\begin{theorem}
  \label{thm:alg2}
  \[Q_E(\f{EXACT}_{k,n-k}^n)\leq n-k+1.
  \]
\end{theorem}
The algorithm we use to prove Theorem~\ref{thm:alg2} in Section~\ref{sec:alg3} involves iteratively applying a unitary, a single query, another unitary, and a measurement. On one hand, the measurement can identify a balanced pair, so we can reduce the problem size. On the other hand, it can either rule out the case $\sum_i x_i = k$, or the case $\sum_{i} x_i = n-k$. We then continue by solving $\f{EXACT}_{n-k}^{n}$ or $\f{EXACT}_{k}^{n}$, respectively.
The first option is favorable, as it quickly decreases the size of the remaining problem. The worst case is using the first query just to decide whether we need to be solving $\f{EXACT}_{n-k}^{n}$ (or the other case). This takes further $n-k$ queries, so the overall number of queries is bounded from above by $1+n-k$.

One might wonder if this algorithm behaves any better than simply first looking at $\f{EXACT}_{k}^{n}$ and if the answer is no, continuing with solving $\f{EXACT}_{n-k}^{n}$. It turns out that the na\"{\i}ve approach is not very efficient, because the algorithm for $\f{EXACT}_k^n$ involves padding the input with extra bits. Imagine $k=\frac{n}{3}$ and $l=\frac{2n}{3}$. To solve $\f{EXACT}_{k}^{n}$, we would need to pad the input with $\frac{n}{3}$ bits and test for $\f{EXACT}^{\frac{4n}{3}}_{\frac{2n}{3}}$. We could be unlucky that $\frac{n}{3}$ queries would just give us unbalanced pairs, always with one useless bit from the padding. After finally learning that there are {\em not} exactly $k$ ones, we could scratch the $\frac{n}{3}$ newly identified $0$'s in the original instance, but we would still have to continue with $\f{EXACT}_{\frac{2n}{3}}^{\frac{2n}{3}}=\f{EQUALITY}_{\frac{2n}{3}}$. This could require another $\frac{2n}{3}$ queries. All in all, in the worst case we would need $n$ queries. However, the algorithm from Section~\ref{sec:alg3} uses at most $n-k+1$ queries, which translates to $\frac{2n}{3}+1$, which is much better.

Second, for the general case $k+l \neq n$, we claim
\begin{theorem}
  \label{thm:alg3}
  \[Q_E(\f{EXACT}_{k,l}^n) \leq \max\{n-k,l\}+1.
  \]
\end{theorem}
\begin{proof}
From Theorem~\ref{thm:alg2} and Lemma~\ref{lem:red}:
    \[Q_E(\f{EXACT}_{k,l}^n)\leq \frac{n+\max\{n-l-k,l+k-n\}+l-k}{2}+1 =\max\{n-k,l\}+1.\]
\end{proof}

\subsubsection{The proof of Theorem~\ref{thm:alg2}: an algorithm for unbalance $\pm d$: }
\label{sec:alg3}

In this Section we prove Theorem~\ref{thm:alg2}, presenting an algorithm for the problem $\f{EXACT}_{k,n-k}^n$ that requires $n-k+1$ queries.

Our goal is to find an algorithm deciding whether the number of 1's in the function values is $k$ or $n-k$. Equivalently, this problem can be also called $\f{UNBALANCE}_{d}^{n}$ with $d=l-k=n-2k$: does the input $x$ have ``unbalance'' $\sum_i \hat{x}_i = \pm d$ or not? This will make it easy to compare with the results of the algorithms in Section~\ref{sec:thealg} for $d=1,2,3$.

We start our algorithm with two registers prepared in the unnormalized state
\[
  \left(\frac{d}{n}\ket{0} + \ket{1}\right)\ket{\S},
\]
with $d$ the unbalance we test for.
Conditioned on the first register being $\ket{1}$, we transform the second register to a uniform superposition of states $\frac{1}{\sqrt{n}}\sum_{i=1}^{n}\ket{i}$. We then query the oracle. This gives us
\[
  \frac{d}{n}\ket{0}\ket{\S} + \frac{1}{\sqrt{n}} \ket{1}\sum_i \hat{x}_i \ket{i},
\] 
Controlled by the first register, we apply the operation $U_n$ from \eqref{Uoperation} to the second register (this is where another factor of $\frac{1}{\sqrt{n}}$ comes from), producing
\[
	\frac{d}{n} \ket{0}\ket{\S} +  
    \frac{1}{n} \ket{1} \sum_i \hat{x}_i \ket{\S} +
    \frac{1}{n} \ket{1} \sum_{i,j\in [n], i<j} (\hat{x}_i-\hat{x}_j) \ket{i,j}.
\]
As we are looking at unnormalized states, we can now omit the common prefactor $\frac{1}{n}$. Finally, we apply a Hadamard\footnote{Observe that the Hadamard operation is equal to $U_2$ \eqref{Uoperation} up to up to relabeling of the basis states.} to the first (ancilla) register and get the unnormalized state
\[
   \left( \left(d + \sum_i  \hat{x}_i\right) \ket{0} + 
    \left(d - \sum_i \hat{x}_i\right) \ket{1} 
    \right) \ket{\S}+
    \left(\ket{0}-\ket{1}\right) 
        \sum_{i,j\in [n], i<j} (\hat{x}_i-\hat{x}_j) \ket{i,j}.
\]
Finally, we measure the second register. Whenever we get a pair $\ket{i,j}$, we know that it is an unbalanced one, with $\hat{x}_i=-\hat{x}_j$. We can get rid of it, and continue solving a smaller problem with $n'=n-2$. On the other hand, if we get $\ket{\S}$ in the second register, we need to look at the ancilla (first) register as well. If the ancilla is $\ket{0}$, we learn that the overall unbalance is not $-d$. On the other hand, if the ancilla is $\ket{1}$, we learn that the overall unbalance is not $d$. Thus, by using a single query, our problem changes from $\f{UNBALANCE}_{d}^{n}$ to $\f{EXACT}_{k}^{n}$ or $\f{EXACT}_{n-k}^{n}$. Switching to the optimal algorithm for $\f{EXACT}_{k}^{n}$, this reduced problem can be solved in $\leq n-k$, i.e. $\leq \frac{n+d}{2}$ queries. 

Therefore, by iterating the above steps, we reduce the problem size by 2 several times, and then at some point reduce the problem to $\f{EXACT}_{k'}^{n'}$ or $\f{EXACT}_{n'-k'}^{n'}$. 
The worst option in terms of the number of queries is when we never reduce the problem size, and use the very first query just to rule out one of the options $d$ or $-d$ for the unbalance.
We then end up having to solve
$\f{EXACT}_k^n$ or $\f{EXACT}_{n-k}^n$, which both can use another $n-k$ queries. Altogether, we require
\[
    Q_E(\f{EXACT}_{k,n-k}^n) \leq n - k + 1
\]
queries. This concludes the proof of Theorem~\ref{thm:alg2}.

For comparison with the algorithms in Section~\ref{sec:thealg}, we can formulate the result in terms of the unbalance $d$. Recalling $d=n-2k$, this algorithm finds the answer using
\[
    n-k+1 = \frac{n+d}{2} + 1.
\]
queries. Recall that we have $l=n-k$ here. 
For $d=1$, this gives $\frac{n+3}{2}$ queries, i.e. one extra query in comparison to Theorem~\ref{thm:alg}.
For $d=2$, this algorithm needs $\frac{n+4}{2}$ queries, i.e. two more queries over Theorem~\ref{thm:alg}. 
For $d=3$, this algorithm needs $\frac{n+5}{2}$ queries, i.e. again two extra queries above Theorem~\ref{thm:alg}. 
Thus, this algorithm is not optimal for these cases and provides just an upper bound. Nevertheless, it works for general $d$, and thus for general $k$. Furthermore, by padding the input, we can get a fully general algorithm (for $k$ and $l$ not tied by $l=n-k$) as described in the proof of Theorem~\ref{thm:alg3}.


\section{Conclusion}
We have shown that the exact quantum query complexity for $\f{EXACT}_{k,l}^n$ is 
\[Q_E(\f{EXACT}_{k,l}^n)=\begin{cases}
    \max\{n-k,l\}\text{, if }d=1\text{ and }l=n-k,\\
    \max\{n-k,l\}-1\text{, if }d\in\{2,3\}.
    \end{cases}\]
where $d=l-k$.
When $d=2$ and $l=n-k$, this provides another example of a symmetric function with $D(f)=2Q_E(f)$ which is the largest known gap between $D(f)$ and $Q_E(f)$ for symmetric functions $f$. To show that $Q_E(\f{EXACT}_{k,k+1}^{2k+1})>k$ we use an approach based on representation theory. We do not know if this lower bound method is sufficient to prove $Q_E(f)\geq \frac{n}{2}$ for all symmetric $f$. In particular, it seems difficult to apply it for the symmetric function $\f{SYM}_a$ has, for example, $a=0^51^50^51^50^5$.

We also give a general algorithm and a lower bound,  for all $l,k$, showing that:
\[\max\{n-k,l\}-1 \leq Q_E(\f{EXACT}_{k,l}^n) \leq \max\{n-k,l\}+1.\]

Previously known quantum algorithms for symmetric functions (e.g., the well known algorithm for PARITY and the algorithms for $\f{EXACT}_k^n$ \cite{AIS13}) typically measure the quantum state after each query.
In contrast, our algorithm for $d\in\{1,2,3\}$ does not have this structure.
Morerover, our numerical simulations suggest that there is no algorithm for $\f{EXACT}_{k,l}^n$ that uses an optimal number of queries and measures the state completely after each query. We think that it is an interesting problem to study 
the power of quantum algorithms with the restriction that after each query the state must be measured completely
and the limits of what can be achieved with such algorithms.

\section*{Acknowledgements}
This research was supported by the ERC Advanced Grant MQC, Latvian State Research Programme NexIT Project No. 1, EU FP7 project QALGO, the People Programme (Marie Curie Actions) EU's 7th Framework Programme under REA grant agreement No. 609427, Slovak Academy of Sciences, and the Slovak Research and Development Agency grant APVV-14-0878 QETWORK. We also thank Bujiao Wu (\email{wubujiao@ict.ac.cn}) for spotting the inaccuracies in the proofs of Lemmas \ref{lem:rs} and \ref{lem:1or4}.


\bibliographystyle{splncs03}

\phantomsection
\addcontentsline{toc}{chapter}{References}
\bibliography{quantum}

\appendix
\section{Proof of Lemma \ref{lem:1or4}}
\begin{proof}
The algorithm is similar to the subroutine $\sr{UNB-R}_d^n$. The goal is to construct an amplitude that is a symmetric polynomial of degree 3 of the form $\sum_{i<j<k} \hat{x}_i\hat{x}_j\hat{x}_k+c\sum_i \hat{x}_i$ that is 0 when $|x|\in\{1,4\}$. Start with
\[
\sum_{i\in[n]}{\hat{x}_i}\ket{\S}+c_1\sum_{\substack{i,j\in[n]\\i<j}}{(\hat{x}_i-\hat{x}_j)\ket{i,j}}.
\]
and perform $R_{\alpha}$ for suitable $\alpha$ on $c_1$ part of the state. Then $U_5^{-1}$ to obtain $c_2$ part of the state and $U_3^{-1}$ on the remainder to obtain $c_4$ part of the state.
\[
c_2\sum_{i\in[n]}
\Bigg(
\sum_{\substack{j\in[n]\\j\neq i}}{\hat{x}_j}
+c_3\hat{x}_i
\Bigg)
\ket{i}\ket{\S}
+c_4\sum_{k\in[n]} \sum_{\substack{i,j\in[n]\backslash\{k\}\\i<j}}
(\hat{x}_i-\hat{x}_j)\ket{k}\ket{i,j}
\]
The constraints induced by these transformations are
    \[(c_2\cdot c_3+4\cdot c_2)/\sqrt{5} = 1, \qquad \qquad \qquad \qquad    
	c_1^2 = (c_2\cdot (c_3-1)/\sqrt{5})^2+(3\cdot c_4/\sqrt{3})^2.\]
We continue with a query on variable indicated by the first register.
\[
c_2\sum_{i\in[n]}{
    \Bigg(
    \hat{x}_i\sum_{\substack{j\in[n]\\j\neq i}}{\hat{x}_j}+c_3
    \Bigg)
    \ket{i}\ket{\S}}+c_4\sum_{k\in[n]}{\hat{x}_k\sum_{\substack{i,j\in[n]\backslash\{k\}\\i<j}}{(\hat{x}_i-\hat{x}_j)\ket{k}\ket{i,j}}}\]
Next, we perform $U_5$ on $c_2$ part of the state to obtain $c_5$ and some $c_7$ and $U_3$ on $c_4$ part of the state to obtain some $c_7$ and $c_8$. Perform $R_\alpha^{-1}$ for suitable $\alpha$ to merge $c_7$ states.
\[
    \begin{split}
    c_5
       \Bigg(
       \sum_{\substack{i,j\in[n]\\i<j}}
       \hat{x}_i \hat{x}_j 
       + c_6
      \Bigg) \ket{\S}
    + c_7 \sum_{\substack{i,j\in[n]\\i<j}} (\hat{x}_i-\hat{x}_j)
    \Bigg(
    \sum_{\substack{k\in[n]\\k\notin\{i,j\}}}{\hat{x}_k}
    \Bigg)
    \ket{i,j}+\\
    +c_8\sum_{\substack{i,j,k,l\in[n]\\i<j\\k<l\\\{i,j\}\cap \{k,l\}=\emptyset}}
    (\hat{x}_i-\hat{x}_j)(\hat{x}_k-\hat{x}_l)\ket{i,j,k,l}
\end{split}
\]
The constraints induced by these transformations are
\begin{align*}
 &5\cdot c_2\cdot c_3/\sqrt{5} =c_5\cdot c_6, 
 &c_2\cdot 2/\sqrt{5} = c_5, \\
 &c_7^2 = c_2^2/5+c_4^2/3, 
 &c_4/\sqrt{3} = c_8.
\end{align*}
First we split $\ket{i,j}$ into two parts using $R_\alpha$. Now we perform $U_5^{-1}$ on $\ket{\S}$ as input sum state and $c\ket{i,j}$ for some $c$ as difference states. The output is stored as $\ket{i}\ket{S}$ states. For each $\ket{i,j}$ we apply $U_3^{-1}$ to $\sqrt{c_7-c^2}\ket{i,j}$ as input sum state and $\ket{i,j,k,l}$ as the difference states. Again, it is easier to verify that we obtain the following state by running the algorithm backwards.
\[
\begin{split}
c_9\sum_i
\Bigg(
\hat{x}_i\sum_{j\neq i}{\hat{x}_j}+c_{10}\sum_{\substack{j<k\\j\neq i\\k \neq i}}
    \hat{x}_j\hat{x}_k + c_{11}
\Bigg)
\ket{i}\ket{\S}+\\
+c_{12}\sum_k
    \sum_{\substack{i<j\\i\neq k\\j\neq k}}{(\hat{x}_i-\hat{x}_j)
\Bigg(
\sum_{l\notin\{i,j,k\}}
    \hat{x}_l}+c_{13}\hat{x}_k
\Bigg)
    \ket{k}\ket{i,j}
\end{split}
\]
This step induces the constraints:
	\[ (2\cdot c_9+3\cdot c_9\cdot c_{10})/\sqrt{5} = c_5, \]
	\[ c_9\cdot c_{11}\cdot 5/\sqrt{5} = c_5\cdot c_6, \]
	\[ c_7^2 = (c_9\cdot (1-c_{10})/\sqrt{5})^2+(c_{12}\cdot (2+c_{13})/\sqrt{3})^2, \]
	\[ c_8 = c_{12}\cdot (c_{13}-1)/\sqrt{3}. \]
Perform a query on variable indicated by the first register.
\[\begin{split}c_9\sum_i{
\Bigg(
\sum_{j\neq i}{\hat{x}_j}+c_{10}\hat{x}_i\sum_{\substack{j<k\\j\neq i\\k \neq i}}{\hat{x}_j\hat{x}_k}+c_{11}\hat{x}_i
\Bigg)
\ket{i}\ket{\S}}+\\+c_{12}\sum_k{\sum_{\substack{i<j\\i\neq k\\j\neq k}}{(\hat{x}_i-\hat{x}_j)
\Bigg(
\hat{x}_k\sum_{l\notin\{i,j,k\}}{\hat{x}_l}+c_{13}
\Bigg)
\ket{k}\ket{i,j}}}
\end{split}
\]
Apply $U_5$ to the $\ket{i}\ket{S}$ states as input states and $\ket{\S}$ as the output sum state and part of $\ket{i,j}$ as output difference states. For each $i,j$ run $U_3$ on $\ket{k}\ket{i,j}$ states as input and $\ket{i,j}$ as the output sum state and $\ket{i,j,k,l}$ as the difference states. Notice, that we obtain $\ket{i,j}$ states from two different sources --- as sum states from $\ket{k}\ket{i,j}$ and difference states from $\ket{i}\ket{\S}$. We then merge them using $R_\alpha^{-1}$ for suitable $\alpha$, getting
\[
\begin{split}
c_{14}
\Bigg(
\sum_{i<j<k} \hat{x}_i\hat{x}_j\hat{x}_k
+c_{15}\sum_i \hat{x}_i
\Bigg)
\ket{\S}+\\
+c_{16} \sum_{i<j}
    (\hat{x}_i-\hat{x}_j)
        \Bigg(
        \sum_{\substack{k<l\\\{i,j\}\cap \{k,l\}=\emptyset}}         
        \hat{x}_k \hat{x}_l + c_{17}
        \Bigg) \ket{i,j} + \\
        +c_{18} \sum_{
                \substack{i,j,k,l,m\\
                i<j\\
                k<l\\
                \{i,j\}\cap \{k,l\}=\emptyset\\
                m \notin \{i,j,k,l\}}}
    (\hat{x}_i-\hat{x}_j) (\hat{x}_k-\hat{x}_l) 
    \hat{x}_m \ket{i,j,k,l}.
\end{split}
\]
This step induces the constraints:
	\[ c_{14} = 3\cdot c_9\cdot c_{10}/\sqrt{5}, \]
	\[ c_{14}\cdot c_{15} = (4\cdot c_9+c_9\cdot c_{11})/\sqrt{5}, \]
	\[ c_{16}^2 = (c_{12}\cdot 2/\sqrt{3})^2 + (c_9\cdot c_{10}/\sqrt{5})^2, \]
	\[ c_{18} = c_{12}/\sqrt{3}, \]
	\[ c_{11}-1 = c_{17}\cdot c_{10}, \]
    \[ 3c_{13} = 2c_{17}.\]

The amplitude of the $c_{14}$ part is zero, when $|x|\in\{1,4\}$. The amplitude of the $c_{16}$ part is zero, when $|x|\in\{0,2,3,5\}$. The amplitude of the $c_{18}$ part is zero, when $|x|\in\{0,1,4,5\}$. This induces constraints:
\[-2 + c_{15}\cdot 3 = 0, \qquad \qquad \qquad \qquad -1+c_{17} = 0. \]
There is only one solution to the system of constraints up to some $\pm$ signs.
\begin{align*}
c_1 &= 1/(4 \sqrt{7}), & c_{10} &= 5,\\
 c_2 &= 17/(16 \sqrt{5}), & c_{11} &= 6,\\
   c_3 &= 12/17, & c_{12} &= (3 \sqrt{3/7})/16,\\
 c_4 &= \sqrt{3/7}/16, &c_{13} &= 2/3,\\
 c_5 &= 17/40, &c_{14} &= 3/8,\\
 c_6 &= 30/17, &c_{15} &= 2/3,\\
   c_7 &= (2 \sqrt{2/7})/5, &c_{16} &= 1/(2 \sqrt{7}),\\
 c_8 &= 1/(16 \sqrt{7}), &c_{17} &= 1,\\
 c_9 &= 1/(8 \sqrt{5}), &c_{18} &= 3/(16 \sqrt{7}).
\end{align*}
\end{proof}

\section{An asymptotically optimal algorithm for a class of symmetric functions}
\label{sec:symalg}


\begin{proof}[Proof of Theorem~\ref{thm:sym}]
The lower bound follows trivially from von zur Gathen's and Roche's work \cite{vZGR97} and the polynomial method \cite{BBC+98}.

For the upper bound, we now present an algorithm. Its main idea is to use the algorithm for $\f{UNBALANCE}_d^n$ from Section~\ref{sec:alg3} to successively eliminate the weights $i$ from consideration, such that $a_i=1$, or reduce the problem size by 2 bits. Once we have eliminated all possible weights $i:a_i=1$ except one, we switch to the algorithm for $\f{EXACT}_k^n$. Thus the algorithm consists of two stages.

In the first stage of the algorithm we will keep track of the problem, to which we have reduced the original problem, by $\mathtt{a}\in\{0,1,*\}^{n'+1}$, where $\mathtt{a}_i=*$ indicates $|x'|\neq i$ where $x'$ is the padded input.
Padding input $x$ with a zero corresponds to appending 0 at the end of $\mathtt{a}$; padding with a one corresponds to prepending 0 at the beginning of $\mathtt{a}$. Let us denote this operation by $\mathtt{a}:0$ or $0:\mathtt{a}$. A useful operation, when we learn that $x_i\neq x_j$, is the removal of first and last elements of $\mathtt{a}$ and it will be denoted by $\downarrow\mathtt{a}\downarrow$. The algorithm starts with $\mathtt{a}=a$.
\begin{algorithm}
\caption{Stage 1: Eliminate multiple weights $i:a_i=1$}\label{alg:elim}

\begin{algorithmic}
\For {$i=1,2,\ldots, 2g(n)$}
    \State $\mathtt{a} \, \gets \,  0:\mathtt{a}$
\EndFor
\State{$\ell \, \gets \, \frac{n}{2}+g(n)$} \Comment{$\ell$ points to the position of leftmost potential 1}
\For {$m=0,\frac{1}{2},1,\ldots,2g(n)-\frac{1}{2},2g(n)$} \Comment{$m$ is the middle of $\mathtt{a}$ relative to $\ell$}
    \While {$\exists \mathtt{a}_i=\mathtt{a}_j=1:\frac{i+j}{2}=\ell+m$}
        \State Run one step of the $\f{UNBALANCE}_{j-i}^{2\ell+2m}$ algorithm from Section~\ref{sec:alg3}
        \If{we learn $x_k\neq x_l$}
            \State $\mathtt{a} \, \gets \,  \downarrow\mathtt{a}\downarrow$
            \State $\ell \, \gets \,  \ell-1$
        \Else \Comment {We learn that $|x'|\neq \mathtt{a}_i$}
            \State $\mathtt{a}_i \, \gets \,  *$
        \EndIf
    \EndWhile
    \State $\mathtt{a} \, \gets \,  \mathtt{a}:0$
\EndFor
\end{algorithmic}
\end{algorithm}
If at any point during this computation $\mathtt{a}$ does not contain $0$ (or $1$, respectively), we output $1$ ($0$, resp.). If we never reach the point where $\mathtt{a}$ does not contain $0$ or $1$, but the first stage finishes, in the end $\mathtt{a}$ must contain exactly one $1$. This is so, because the pointer to the middle $m$ sweeps through all 1-s originally in $a$.

Now we proceed with the second stage of our algorithm. To decide $\mathtt{a}$, knowing that for exactly one weight $i:\mathtt{a}_i=1$, we use an algorithm for $\f{EXACT}_{k'}^{n'}$ where $n'$ is the size of the input at the end of the first stage and $k'$ is the only weight with $\mathtt{a}_{k'}=1$.

Let us calculate the number of queries used by the first stage.
\[\#(\text{queries used by Stage 1})=\#(\text{we hit branch }|x|\neq a_i)+\#(\text{we hit branch }x_k\neq x_l).\]
Since the initial $a$ contains at most $2g(n)+1$ ones, $\#(\text{we hit branch }|x|\neq a_i)\leq 2g(n)$. For the sake of brevity, let us denote $\#(\text{we hit branch }x_k\neq x_l)$ by $t$. The number of queries used by second stage of the algorithm depends on $t$, because we have eliminated $2t$ variables from $x$, padded or otherwise.

Now, to determine the complexity of the second stage, we calculate $n'=n+6g(n)+1-2t$ and $\frac{n}{2}+g(n)-t\leq k'\leq \frac{n}{2}+3g(n)-t$. Then,
\[
\#(\text{queries used by Stage 2})= Q_E(\f{EXACT}_{k'}^{n'}) \leq \max\left\{\frac{n}{2}+5g(n)+1-t,\frac{n}{2}+3g(n)-t\right\}.
\]
The total number of queries we use is at most
\[\begin{split}\#(\text{queries used by Stage 1})+\#(\text{queries used by Stage 2})\leq \\
\leq 2g(n)+t+\frac{n}{2}+5g(n)+1-t = \frac{n}{2}+7g(n)+1=\frac{n}{2}+o(n)\sim\frac{n}{2} \end{split}.\]

\end{proof}
The algorithm from Section~\ref{sec:alg3} can be extended to either exclude one of two weights for $\f{EXACT}_{k,l}^n$ or two opposite inputs for more general $k,l$, as described in Section~\ref{sec:unbalanceUW}. The only requirement for the weights $k,l$ imposed by the algorithm is that $k<\frac{n}{2}<l$. Using that algorithm, the constant $c$ in front of $g(n)$ can be decreased from 7 to 5. The constant is relevant when $g(n)=\epsilon n$ since it lets us construct exact quantum algorithms using less than $n$ queries, provided $\epsilon < \frac{1}{2c}$.


\section{Another algorithm: testing for unbalance $+u,-w$}
\label{sec:unbalanceUW}

In this Section, we develop a test for two particular, nonzero values of unbalance with opposite signs, i.e. solving the problem $\f{EXACT}_{\frac{n-u}{2},\frac{n+w}{2}}^{n}$ for $0<u,w\leq n$.
\dnote{Some care needs to be taken if $u$ or $w$ is zero, we can then still flip the function and continue\dots}
It can be later used quite efficiently to test for a range of unbalances, eliminating their extreme values one by one (from either end, at random), or decreasing the problem size.
The test generalizes the approach of Section~\ref{sec:generalalg}. 

\subsection{Testing for two particular values $\{+u,-w\}$ of the unbalance.}

Our goal is to determine whether for a function with exactly $k$ ones, the unbalance $\sum_i \hat{x}_i = -k + n-k = n-2k$ is exactly $+u$ or $-w$, with $u,w> 0$. We will to this by an algorithm that repeats a number of steps that either end up reducing the problem size, eliminating the option $+u$, or eliminating the option $-w$.

Let us start with an unnormalized, two-register (ancilla, data) state
\[
	\left(\frac{\sqrt{uw}}{n} \ket{0} + \ket{1}\right) \ket{\mathcal{S}}.
\]
Controlled on the first, ancilla register being $\ket{1}$, we prepare a uniform superposition in the second, data register and query the data register. We obtain
\[
		\frac{\sqrt{uw}}{n}\ket{0} \ket{\mathcal{S}} 
        + \frac{1}{\sqrt{n}} \ket{1} \sum_i \hat{x}_i \ket{i}
\]
Next, again controlled on the ancilla being $\ket{1}$, we apply the unitary $U_n$ \eqref{Uoperation}, giving us
\[
	\frac{\sqrt{uw}}{n} \ket{0}\ket{\mathcal{S}} 
    + \frac{1}{n} \left(\sum_{i} \hat{x}_i\right) \ket{1} \ket{\mathcal{S}}
    + \frac{1}{n}\ket{1} \sum_{i<j} \left(\hat{x}_i-\hat{x}_j\right) \ket{i,j}.
\]
We look at unnormalized states, so we can drop the overall normalization factor $\frac{1}{n}$.
Finally, conditioned on the second register being $\ket{\mathcal{S}}$, we rotate the first qubit using the unitary
\[
	Q =  \frac{1}{\sqrt{u+w}} \left[
		\begin{array}{rr}
			\sqrt{u} & - \sqrt{w}\\\sqrt{w} & \sqrt{u}
		\end{array}
	\right], 
\]
This results in the (unnormalized) state
\[
    \frac{1}{\sqrt{u+w}} \left(
        \sqrt{w} \left( u - \sum_i \hat{x}_i \right) \ket{0}
    + \sqrt{u} \left( w + \sum_i \hat{x}_i \right) \ket{1}
    \right)\ket{\mathcal{S}}
    + \ket{1} \sum_{i<j} \left(\hat{x}_i-\hat{x}_j\right) \ket{i,j}.
\]
Finally, we perform a full measurement. Whenever we get a pair $\ket{i,j}$ in the second register, we know that it is unbalanced, with $\hat{x}_i=-\hat{x}_j$. We can get rid of it, and continue solving a smaller problem with $n'=n-2$, and the same possible unbalances. On the other hand, if we the second register is  $\ket{\mathcal{S}}$, the value of the first register tells us that the unbalance is either not equal to $u$ (if we measure $\ket{0}$, or not equal to $-w$ (if we measure $\ket{1}$).

\dnote{What now? It seems like we can't make a test which tests {\em only} for a value $w$, or, more generally, for two values of unbalance with the same sign. For that, we probably need to append the system with a few more known bits, to shift one of the tested unbalances on the other side of zero.

This is still not optimal (I believe) as we need one step to break the symmetry, and then we use the optimal EXACT$^n_k$ subroutine to figure out the rest. On the other hand, it at least achieves some speedup for the class of functions that J\=anis describes.
}

\subsection{Computing $\f{SYM}_a$}
We can use the algorithm described in this appendix (instead of the algorithm from Section~\ref{sec:alg3}) as a subroutine to the algorithm Appendix~\ref{sec:symalg}.  We also slightly modify the algorithm from Appendix~\ref{sec:symalg}. Now instead of moving $m$ through all of the 1-s in $a$, we start with $l+m=\frac{n}{2}$. Then call algorithm from Appendix~\ref{sec:unbalanceUW} until either 1-s left of $m$ are eliminated or 1-s to the right of $m$ are eliminated. Then proceed by moving $m$ in the direction of the remaining 1-s.

Again, the algorithm has two stages with complexities:
\[\#(\text{queries used by Stage 1})=\#(\text{we hit branch }|x|\neq a_i)+\#(\text{we hit branch }x_k\neq x_l);\]
and if we denote $\#(\text{we hit branch }x_k\neq x_l)$ in the Stage 1 by $t$, the number of variables in the input after Stage 1 by $n'=n+2g(n)-2t$ and by $k'$ the only weight for which $\mathtt{a}_{k'}=1$, $\frac{n}{2}-g(n)-t \leq k'\leq \frac{n}{2}+3g(n)-t$ then
\[
\#(\text{queries used by Stage 2})= Q_E(\f{EXACT}_{k'}^{n'}) \leq \frac{n}{2}+3g(n)-t.\]
Thus the total number of queries used is $\leq \frac{n}{2}+5g(n)$.

\end{document}